\documentclass{article}

\usepackage{arxiv}
\usepackage{amsmath,amssymb,amsfonts,amsthm}
\usepackage{algorithmic}
\usepackage{graphicx,float}
\usepackage[ruled,vlined]{algorithm2e}
\usepackage[utf8]{inputenc}
\usepackage[english]{babel}
\usepackage{xcolor}
\usepackage{cite}
\usepackage{outlines}
\usepackage[hidelinks]{hyperref}       
\usepackage{graphicx}
\graphicspath{{fig/}}
\hypersetup{
		colorlinks=true,
		citecolor=blue,
		linkcolor=blue,
		filecolor=magenta,      
		urlcolor=cyan,
}
\usepackage[nameinlink,noabbrev]{cleveref}

\DeclareMathOperator{\bmp}{\otimes}

\title{Boolean Matrix Factorization via Nonnegative Auxiliary Optimization}

\author{ 
	Duc P. Truong \\
	Computer, Computational and Statistics Division\\
	Los Alamos National Laboratory, USA\\
	\texttt{dptruong@lanl.gov} \\
	\And
	Erik Skau \\
	Computer, Computational and Statistics Division\\
	Los Alamos National Laboratory, USA\\
	\texttt{email} \\
	\And
	Derek Desantis \\
	Theoretical Division\\
	Los Alamos National Laboratory, USA\\
	\texttt{email} \\
	\And
	Boian Alexandrov \\
	Theoretical Division\\
	Los Alamos National Laboratory, USA\\
	\texttt{email} \\
}
\renewcommand{\shorttitle}{BANMF}

\begin{document}
\maketitle

\begin{abstract}
	A novel approach to Boolean matrix factorization (BMF) is presented. 
	Instead of solving the BMF problem directly, this approach solves a nonnegative optimization problem with the constraint over an auxiliary matrix whose Boolean structure is identical to the initial Boolean data. Then the solution of the nonnegative auxiliary optimization problem is thresholded to provide a solution for the BMF problem. We provide the proofs for the equivalencies of the two solution spaces under the existence of an exact solution. Moreover, the nonincreasing property of the algorithm is also proven. Experiments on synthetic and real datasets are conducted to show the effectiveness and complexity of the algorithm compared to other current methods.
\end{abstract}

\newcommand{\XBij}{(X_B)_{ij}}
\newtheorem{theorem}{Theorem}
\newtheorem{example}{Example}
\newcommand{\B}{\mathbb{B}}
\newcommand{\R}{\mathbb{R}}
\newcommand{\rankb}{\operatorname{rank}_{\B}}
\newcommand{\ranknn}{\operatorname{rank}_{+}}
\newcommand{\rank}{\operatorname{rank}}
\newcommand{\support}{\operatorname{support}}
\newcommand{\BMF}{BANMF}
\renewcommand{\algorithmautorefname}{Algorithm}

\section{Introduction}
Many research fields, such as personalized medicine, space research, social sciences, climate research, nonproliferation, emergency response, finances, etc. collect and generate large-scale datasets. Analysis of such datasets is difficult, since often the underlying fundamental processes (or features) remain hidden or latent (i.e., not directly observable) \cite{everett2013introduction}. 
Extracting such latent features reveals valuable information and hidden causality and relations. One of the most powerful tools for extracting latent (hidden) features from data is factor analysis \cite{spearman1961general}. 
Traditionally, factor analysis approximates a data-matrix $X \in \R^{n \times m}$ by a product of two low-rank (factor) matrices, $X \approx WH$, where $W \in \R^{n \times k}$, $H \in \R^{k \times m}$, and $k \ll n,m$. Various types of factorizations can be obtained by imposing different constraints. For instance, imposing orthogonality on the factors leads to the singular value decomposition (SVD) \cite{Stewart1993}, while the nonnegativity constraint leads to non-negative matrix factorization (NMF) \cite{lee1999learning}. In a lot of investigations, the variables are simple dichotomies, \{\textit{false}, \textit{true}\}, that is, the data contains only binary values, $\{0, 1\}$. For example, in relational databases, an object-attribute relation is represented by a Boolean variable, which takes value $\{1\}$-\textit{true}, if the object has this attribute, or $\{0\}$-\textit{false}, otherwise. In this case, all the values of the factors have to be $0$ or $1$. In this case, instead of simple arithmetic we also need to use Boolean algebra: instead of “plus” and “times” we need the logical operations of “or” with “and”, which results in a Boolean factorization \cite{miettinen2008discrete}. To be precise, letting $\B=\{0,1\}$,
and assuming the input data is a matrix  $X \in \B^{n \times m}$, where $\B=\{0,1\}$, Boolean matrix factorization (BMF) is looking for two binary matrices $W \in \B^{n \times k}$, $H \in \B^{k \times m}$ such that $X \approx W \otimes_B H$, where $(W \otimes_B H)_{ij}=\vee_{l=1}^k W_{il} \land H_{lj} \in \B$ and $\vee$ and $\land$ are the logical "or" and "and" operations. The \textit{Boolean rank} of $X$ is the smallest $k$ for which such an exact representation, $X=W \otimes_B H$, exists. Interestingly, in contrast to the non-negative rank, the Boolean rank can be not only bigger or equal, but also smaller than the real rank \cite{Monson1995, desantis2020factorizations}. 

In this paper, we propose a new algorithm for BMF, we call Boolean Auxiliary Non-negative Matrix Factorization (\BMF{}).  Here, we consider  a specific constrained optimization over an auxiliary matrix $Y$ related to the initial binary data matrix $X$. By performing the constrained optimization over the auxiliary matrix we are able to simulate a Boolean matrix product by enforcing the support of the matrix product and disregarding the exact matrix values. From the solution of our relaxed problem we apply thresholding to our solution matrices $W$ and $H$ to arrive at the BMF solution. We demonstrate that our method outperforms other state-of-the-art methods for the BMF problem when there exists a discrepancy between the nonnegative and Boolean rank.  Moreover, \BMF{} has superior performance under corruption by Bernoulli noise and at various density levels.  Furthermore \BMF{} is shown to have faster run times than PNL-PF, a competitor algorithm, for large datasets.

\section{Related Work}

The BMF problem is NP-hard \cite{miettinen2008discrete}, leading to the development of a variety of strategies to approximate a solution. Greedy algorithms were some of the first proposed methods to solving BMF \cite{miettinen2008discrete, belohlavek2010discovery, lucchese2010mining}. There is also a family of Bayesian methods. For example, the authors in \cite{rukat2017bayesian} use a probabilistic generative model and derive a sampling method to accelerate the optimization for BMF. In \cite{ravanbakhsh2016boolean} they use a Bayesian framework to consider the BMF problem as a maximum log likelihood problem and use a message passing procedure to approximate the MAP assignment.

An alternative approach to BMF is to relax the Boolean constraint into a non-negativity constraint.  Several different methods to relax BMF to a NMF problem exist. The thresholding method in \cite{zhang2007binary} attempts take an NMF solution to a binary matrix factorization solution.  This idea can easily be modified to search for a BMF solution by searching for a optimal thresholding values. The Post Nonlinear Penalty Function Algorithm (PNLPF) \cite{miron178boolean} attempts to improve the relaxation by putting a nonlinear function on the product WH. The nonlinear Heaviside function clips arbitrarily large numbers in the product WH back to the [0,1] interval. This composition of regular matrix product and the nonlinear function approximates a Boolean matrix product resulting in a BMF estimation algorithm.

\section{Proposed Approach}
Given a Boolean matrix $X \in \B^{M \times N}$, the BMF optimization problem is
\begin{equation}
\begin{aligned}
& \underset{W,H}{\text{minimize}}
& & ||X-W \bmp_{B} H||_F \\
& \text{subject to}
& & W \in \B^{M \times k},\\
& & & H \in \B^{k \times N},
\end{aligned}
\label{boolean_nmf}
\end{equation}
where \(\bmp_{B}\) denotes \textit{Boolean matrix product } $(W \bmp_B H)_{ij} = \lor_{l=1}^k W_{il} \land H_{lj}$, and $||...||_F$ is the Frobenius norm.

\subsection{Nonnegative auxiliary optimization for BMF} 
Here we propose an alternative nonnegative optimization problem  to solve the BMF problem~(\ref{boolean_nmf}) which is comprised of two steps.  The first step is based on NMF optimization with additional constrains on an auxiliary matrix $Y$.  Given a solution to the first step, the second step consists of thresholding to obtain a Boolean factors.  

We now describe the nonnegative auxiliary optimization problem.  Given a binary matrix $X$, we search for a rank \(k\) Boolean decomposition by solving the optimization,
\begin{equation}
\begin{aligned}
& \underset{Y,W,H}{\text{minimize}}
& & ||Y-W H||_F \\
& \text{subject to}
& & 1 \leq Y_{i,j} \leq k, \text{ if } X_{i,j} = 1 \\
& & & Y_{i,j} = 0, \text{ if } X_{i,j} = 0 \\
& & & W \in \R_+^{N \times k},\\
& & & H \in \R_+^{k \times M}.
\end{aligned}
\label{prob:NAO}
\end{equation}
The solution $(Y,W,H)$ of problem~(\ref{prob:NAO}) is not necessarily Boolean. For the second step, we convert the non-negative $(W,H)$ to Boolean matrices $(\hat{W},\hat{H})$ according to a thresholding function.  Given threshold parameter $\delta \geq 0$, we define
\begin{equation}
	\hat{C}_{ij} = 
\begin{cases}
	\begin{aligned}
		1 & \quad \text{ if } C_{ij} > \delta\\
		0 & \quad \text{ if } C_{ij} \leq \delta
	\end{aligned}
\end{cases}
\label{eqn:Booleanization}
\end{equation}
In theory, we use $\delta = 0$.  However as will be discussed in Algorithm \ref{alg:booleanization}, $\delta >0$ in practice. 

Our method can be viewed as an NMF approach for BMF problem similar to the other relaxation approaches in the literature~\cite{zhang2007binary, miron178boolean}. However, by adding the auxiliary matrix constraint, the \BMF{} approach forces the support of NMF to adhere to that of a BMF while disregarding the exact values.  The subsequent results establish this.  Under the existence of an exact decomposition, below we prove the equivalencies of the two solution spaces. This will be done with two theorems, one theorem showing containment in each direction. The following theorem proves that an exact decomposition solution of BMF problem (\ref{boolean_nmf}) is also a solution of our relaxed nonnegative optimization problem (\ref{prob:NAO}).

\begin{theorem}
Given an exact decomposition solution $(W,H)$ to problem~(\ref{boolean_nmf}) such that $X = W \bmp_B H$, there exists $Y \in \R_+^{N \times M}$ such that $(Y,W,H)$ is an exact decomposition solution for the optimization problem (\ref{prob:NAO}).
\label{theorem:bmf to mc}
\end{theorem}

\begin{proof}

Given an an exact solution $(W, H)$ such that $X = W \bmp_B H$, we construct a \(N \times M\) nonnegative matrix $Y = WH$.

Now we need to show that $(Y, W, H)$ is an exact solution for the optimization problem (\ref{prob:NAO}). First, by construction, the optimized value is already zero. Next, we need to show that the solution satisfies the constraints.\\
(a) \(W \in \B^{N \times k} \Rightarrow W \in \R_+^{N \times k}\)\\
(b) \(H \in \B^{k \times M} \Rightarrow H \in \R_+^{k \times M}\)\\
(c) For \((i,j) \notin \operatorname{support}(X)\), which means \(X_{ij} = 0\)
\(\Rightarrow W_{il} = H_{lj} = 0 \ \forall l = 1,\dots,k.\) by the definition  of Boolean matrix product in \autoref{boolean_nmf}\\
\(\Rightarrow Y_{ij} = (WH)_{ij} = \sum_{l=1}^k (W)_{il}(H)_{lj} = 0\)\\
(d) For \((i,j) \in \operatorname{support}(X)\), which means \(X_{ij} = 1\)
\(\Rightarrow \exists l \in \{1,\dots,k\}\) such that \(W_{il} = H_{lj} = 1\). Since \(W_{i,:},\ H_{:,j} \in \mathbb{R}^{k}\), we have \(1 \leq (WH)_{ij} \leq k \Rightarrow 1 \leq Y_{ij} = (WH)_{ij} \leq k \)

Therefore, from (a),(b),(c), (d) and the fact that the optimized value is zero, we have shown that \(Y=WH\) is an exact solution for the optimization problem (\ref{prob:NAO}) as needed.
\end{proof}

This proves that the BMF solution set is contained in the \BMF{} solution set when there exists an exact decomposition. The next theorem proves containment the other direction, that the \BMF{} solution set, after Booleanizing the solutions, is contained in the BMF solution set when there exists an exact decomposition.

\begin{theorem}
	Given an exact solution $Y = WH$ for the \BMF{} problem (\ref{prob:NAO}), and \(\hat{W}\) and \(\hat{H}\) are the Booleanized matrices of \(W\) and \(H\) with threshold $\delta = 0$, then $X = \hat{W} \bmp_B \hat{H}$ is an exact solution for the BMF problem (\ref{boolean_nmf}).
	\label{theorem:mc to bmf}

\end{theorem}
\begin{proof}
	Suppose that \(Y\), \(W\), and \(H\) is an exact solution for the \BMF{} problem, and that \(\hat{W}\) and \(\hat{H}\) are the Booleanized matrix of \(W\) and \(H\) with threshold $\delta = 0$. We show that \(X = \hat{W} \bmp_B \hat{H}\).\\
	First, following the logic in the proof of \autoref{theorem:bmf to mc}, we have that:\\
	\begin{equation}
		(\hat{W}\bmp_B \hat{H})_{ij} = 
		\begin{cases}
			\begin{aligned}
				&0 &\text{ if } (WH)_{ij} = 0 \\
				&1 &\text{ if } (WH)_{ij} > 0
			\end{aligned}
		\end{cases}
	\end{equation}
	Second, by the constraint of the optimization and the assumption that \(Y=WH\), we also have:\\
	\begin{equation}
		X_{ij} = 
		\begin{cases}
			\begin{aligned}
				&0 &\text{ if } Y_{ij} = (WH)_{ij} = 0 \\
				&1 &\text{ if } Y_{ij} = (WH)_{ij} > 0
			\end{aligned}
		\end{cases}
	\end{equation}

	These two above facts imply that \(X = \hat{W} \bmp_B \hat{H}\) or that \(\hat{W}\) and \(\hat{H}\) is an exact solution for the BMF problem. 
\end{proof}

\subsection{Algorithm}
Our implementation of the \BMF{} algorithm cycles between updating the NMF factors $W$, $H$, and the corresponding auxiliary matrix $Y$.  We use the multiplicative update method introduced by Lee and Seung \cite{lee2001algorithms} for the $W$ and $H$ updates, but any other alternating NMF algorithms are also suitable.  For the update of the auxiliary variable $Y$, we  project $Y$ onto the feasible set defined by the constraints. \autoref{alg:NMC} is an outline of a procedure to solve the \BMF{} problem (\ref{prob:NAO}). To find the appropriate threshold parameters $\delta_W$ and $\delta_H$, and the corresponding Boolean matrices, $\hat{W}$ and $\hat{H}$, we employ \autoref{alg:booleanization}.

\begin{algorithm}
	\SetAlgoLined
	\KwIn{Boolean matrix \(X \), latent dimension \(k\), maximum iteration \(Niter\) }
	\KwResult{\(Y, W, H\)}
	Initialization: \(W = rand(N,k)\), \(H = rand(k,M)\), \(Y = X\) \\
	
	\lWhile{iter \(\leq\) Niter}
	{\par
		(1) Update \(W\) using NMF multiplicative update \\
		\(W_{ij} = W_{ij} \dfrac{(YH^T)_{ij}}{(WHH^T)_{ij}}\)\\
		(2) Update \(H\) using NMF multiplicative update \\
		\(H_{ij} = H_{ij} \dfrac{(W^TY)_{ij}}{(W^TWH)_{ij}}\)\\
		(3) Update \(Y_{ij}\) for \((i,j) \in support(X)\) \\
		\(Y_{ij} = \begin{cases}
			1 \quad \text{if} \quad (WH)_{ij} < 1 \\
			(WH)_{ij} \quad \text{if} \quad 1 \leq (WH)_{ij} \leq k \\
			k \quad \text{if} \quad (WH)_{ij} \geq k
		\end{cases}\)\\
		\textbf{endwhile}}
	\caption{\BMF{}}
	\label{alg:NMC}
\end{algorithm}

\begin{algorithm}
	\SetAlgoLined
	\KwIn{Boolean matrix \(X \), nonnegative matrices \(W\) and \(H\), number of points \(npoint\)}
	\KwResult{\(\hat{W}, \hat{H}\)}

	\(W_{point} = \text{linspace(min(W), max(W), npoint)}\)\\
	\(H_{point} = \text{linspace(min(H), max(H), npoint)}\)\\
	\(\delta^*_W, \delta^*_H = \underset{\delta_W \in W_{point}, \delta_H \in H_{point}}{\operatorname{argmin}} |X - \hat{W}^{\delta_W}\bmp\hat{H}^{\delta_H}|\)\\
	\(\hat{W} = \hat{W}^{\delta^*_W}\)\\
	\(\hat{H} = \hat{H}^{\delta^*_H}\)\\
	\caption{Booleanization}
	\label{alg:booleanization}
\end{algorithm}

\subsubsection{Convergence}

We now discuss the convergence of our implementation of the \BMF{} algorithm. We show that Algorithm \ref{alg:NMC} results in a non-increasing updates to the objective function:

\begin{theorem}[Nonincreasing property of \BMF{} algorithm]
	The objective function \(||Y-WH||_F\) in the boolean auxiliary nonnegative matrix factorization (\ref{prob:NAO}) is non-increasing under the \BMF{} update rules.
	\label{theorem:convegence of mc}
\end{theorem}
\begin{proof}
	Here we just need to show that \(||Y_{t-1}-W_{t-1}H_{t-1}|| \geq ||Y_{t}-W_{t}H_{t}||\) at an arbitrary iteration \(t\).\\ 
	As shown in \cite{lee2001algorithms}, the distance \( ||X - WH||_F\) is non-increasing under the multiplicative update rules. Applying this to the updated \(W_t,H_t\)  at the iteration \(t\), we have:
	\begin{equation}
		||Y_{t-1} - W_{t-1}H_{t-1}|| \geq ||Y_{t-1} - W_{t}H_{t}||.
		\label{eqn:conv_mu}
	\end{equation}
	The next step is updating the elements of \(Y_{t-1}\) using the update rule in step 3 of \autoref{alg:NMC}. Note that the update rule minimizes the entry-wise error of $|(Y_t)_{ij} - (W_tH_t)_{ij}|$ within the solution space for $Y$. Indeed, each matrix entry for $Y_t$ either has zero error (for $1 \leq (WH)_{ij} \leq k)$ or has minimal error (for $(WH)_{ij}\leq 1$ and $(WH)_{ij} \geq k$) which is an improvement over $|(Y_{t-1})_{ij} - (W_tH_t)_{ij}|$.  Hence, this step produces \(Y_t\) which further reduces the error
	\begin{equation}
		||Y_{t-1} - W_tH_t|| \geq || Y_t - W_tH_t||
		\label{equn:conv_update}
	\end{equation}
	From inequality (\ref{eqn:conv_mu}) and (\ref{equn:conv_update}), we get \(||Y_{t-1}-W_{t-1}H_{t-1}|| \geq ||Y_{t}-W_{t}H_{t}||\) as needed. Therefore, the objective function of the problem (\ref{prob:NAO}) is non-increasing under \BMF{} update rules.
\end{proof}

\subsubsection{Regularized \BMF{}}
Additionally, to aid in the Booleanization procedure, the factors $W$ and $H$ can be guided to be binary factors with an additional regularization during the optimization. Here we incorporate the $x^2-x$ based regularization used in \cite{miron178boolean, zhang2007binary}. Our regularized problem is formed as follows:
\begin{equation}
	\begin{aligned}
	& \underset{Y,W,H}{\text{minimize}}
	& & ||Y-W H||_F + \dfrac{1}{2} \lambda ||W^2 - W||^2_F \\
	& & & + \dfrac{1}{2} \lambda ||H^2 - H||^2_F\\
	& \text{subject to}
	& & 1 \leq Y_{i,j} \leq k, \text{ if } X_{i,j} = 1 \\
	& & & Y_{i,j} = 0, \text{ if } X_{i,j} = 0 \\
	& & & W \in \R_+^{N \times k},\\
	& & & H \in \R_+^{k \times M}.
	\end{aligned}
	\label{eqn: reg_mc}
\end{equation}
Notice that under the exact solution, the regularization terms will be zeros. Therefore, \autoref{theorem:bmf to mc} and \autoref{theorem:mc to bmf} are also true for the regularized problem (\ref{eqn: reg_mc}). With regularization, the update rules do not have the non-increasing property anymore, but empirically we observe reasonable convergence in practice. \autoref{fig: reg_mc_convergence} shows the convergence of the regularized BANMF on a synthetic dataset. Our implementation of the regularized \BMF{} is described in \autoref{alg:reg_mc}.

\begin{algorithm}
	\SetAlgoLined
	\KwIn{Boolean matrix \(X \), latent dimension \(k\), maximum iteration \(Niter\), regularization parameter \(\lambda\) }
	\KwResult{\(\hat{W}^{(\delta_W)}\), \(\hat{H}^{(\delta_H)}\)}
	Initialization: \(W = rand(N,k)\), \(H = rand(k,M)\), \(Y = X\) \\
	\lWhile{iter \(\leq\) Niter}
	{\\
		(1) Update \(W\) and \(H\) using NMF multiplicative update \\
		\(W_{ij} = W_{ij} \dfrac{(YH^T)_{ij} + 3\lambda W_{ij}^2}{(WHH^T)_{ij} + 2\lambda W_{ij}^3 + \lambda W_{ij}^2 }\) \\  
		\(H_{ij} = H_{ij} \dfrac{(W^TY)_{ij} + 3\lambda H_{ij}^2}{(W^TWH)_{ij} + 2\lambda H_{ij}^3 + \lambda H_{ij}^2}\)\\
		(2) Update \(Y_{ij}\) for \((i,j) \in support(X)\)\\
		\(Y_{ij} = \begin{cases}
			1 \quad \text{if} \quad (WH)_{ij} < 1 \\
			(WH)_{ij} \quad \text{if} \quad 1 \leq (WH)_{ij} \leq k \\
			k \quad \text{if} \quad (WH)_{ij} \geq k
		\end{cases}\)
	
	\textbf{endwhile}}
	Boolean algorithm is applied to identify \(\delta_W\) and \(\delta_H\) for \(\hat{W}^{(\delta_W)}\), \(\hat{H}^{(\delta_H)}\)
	\caption{Regularized \BMF{}}
	\label{alg:reg_mc}
\end{algorithm}
\begin{figure}
	\centering
	\includegraphics[width = 0.45\textwidth]{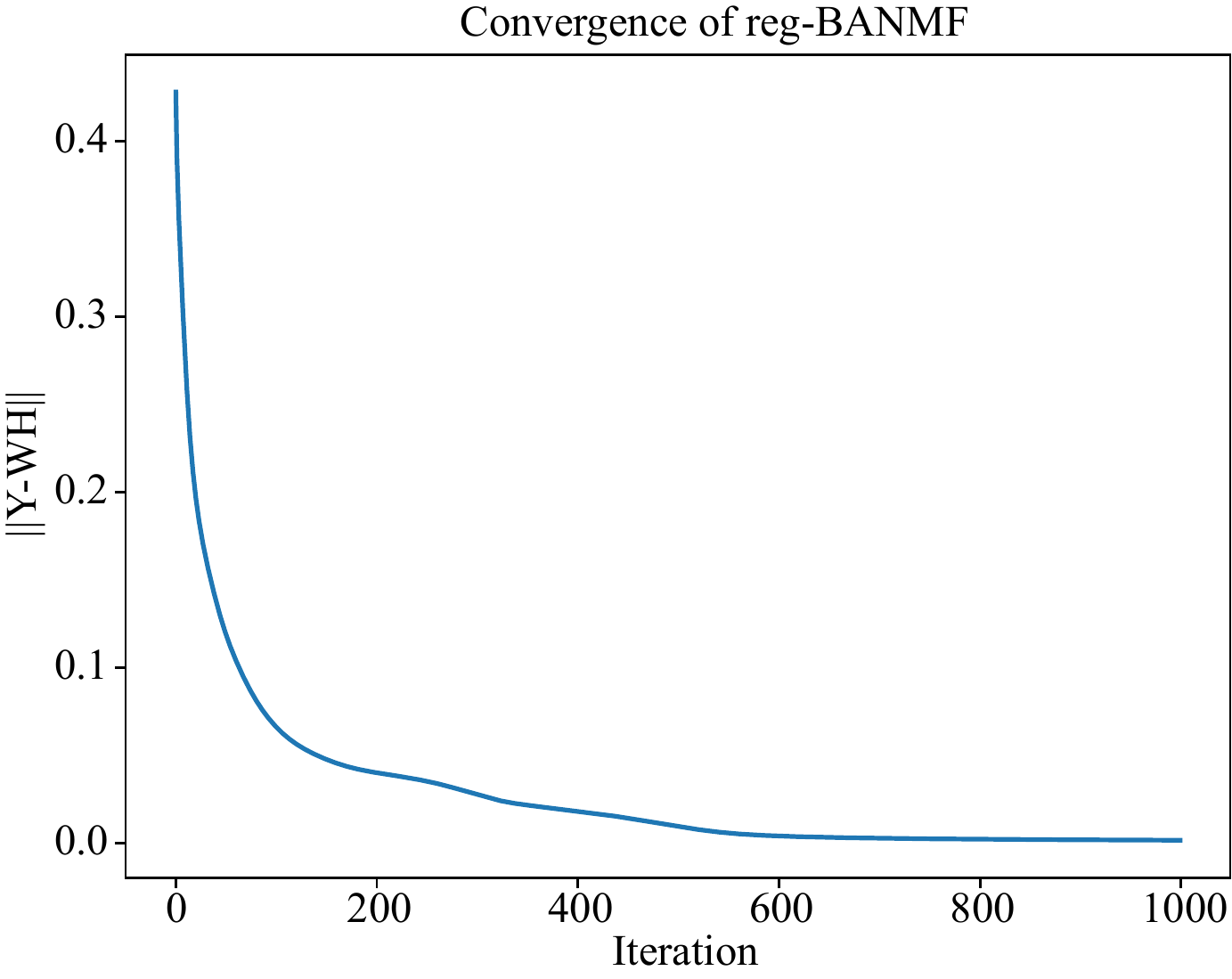}
	\caption{Convergence of regularized BANMF on a synthetic dataset.}
	\label{fig: reg_mc_convergence}
\end{figure}

\section{Experimental Results}

In this section, we experimentally evaluate the performance of the \BMF{} algorithm. For comparison we employ ASSO \cite{miettinen2008discrete}, NMF \cite{lee2001algorithms}, and post nonlinear penalty function (PNL-PF) \cite{miron178boolean} where the solutions are thresholded using Booleanization \autoref{alg:booleanization} to be binary. First, we demonstrate that if the non-negative rank and the Boolean rank are different, then \BMF{}, and in particular Regularized \BMF{}, outperform the other methods.   Next, we investigate how methods perform with regard to different noise and density levels.  We then show that \BMF{} outperforms its strongest contender, PNL-PF, in execution time for random matrices. Lastly, the methods are evaluated when applied to three real datasets: (1) UCI Zoo, (2) Congressional voting records, and (3) Lung cancer with probes.

\subsection{Generating Boolean synthetic data}
\label{generating data}
First, we briefly describe our process for generating Boolean synthetic data. For each experiment, we will generate random matrices 
$$X^{N \times M} = W^{M \times k} \bmp_B H^{ k \times m}$$
where
$$W_{ij}, H_{ij} \sim Bernoulli(p).$$ 
This stochastic process generates data with different features which probabilistically depend on the parameters. In each experiment, this process will be augmented to analyze statistical performance across various data constraints.

To construct a Boolean rank $k$ data matrix with a desired density $d$, we compute the probability that any given entry of $X$ is one, $P(X_{ij}=1) = P(\lor_{l=1}^k W_{il} \land H_{lj} = 1) = 1 - (1-P(W_{il}=1)*P(H_{lj}=1))^k$. By imposing that the densities of $W$ and $H$ are the same, and inverting the probability formula, we arrive with the equation for the density of $W$ and $H$ that correspond to the desired density of $X$, $\sqrt{1 - (1-d)^\frac{1}{k}}$.

\subsection{Different density levels}

In the first experiment, we investigated the performance of the  methods on random Boolean matrices with different density levels.  Here 100 different $50 \times 50$ matrices with $k=5$ were generated for each of three different density levels: $20\%, 50\%$ and $80\%$. 

\autoref{fig:synthetic_Bool_err} shows ASSO performs rather poorly across all density levels. For sparse matrices, BANMF and REG-BANMF perform quite well.  However for dense matrices, PNLPF outperforms every method.  Overall, (regularized) \BMF{} consistently performs well across different sparsity levels.

\begin{figure}
	\centering
	\includegraphics[width = 0.4 \textwidth]{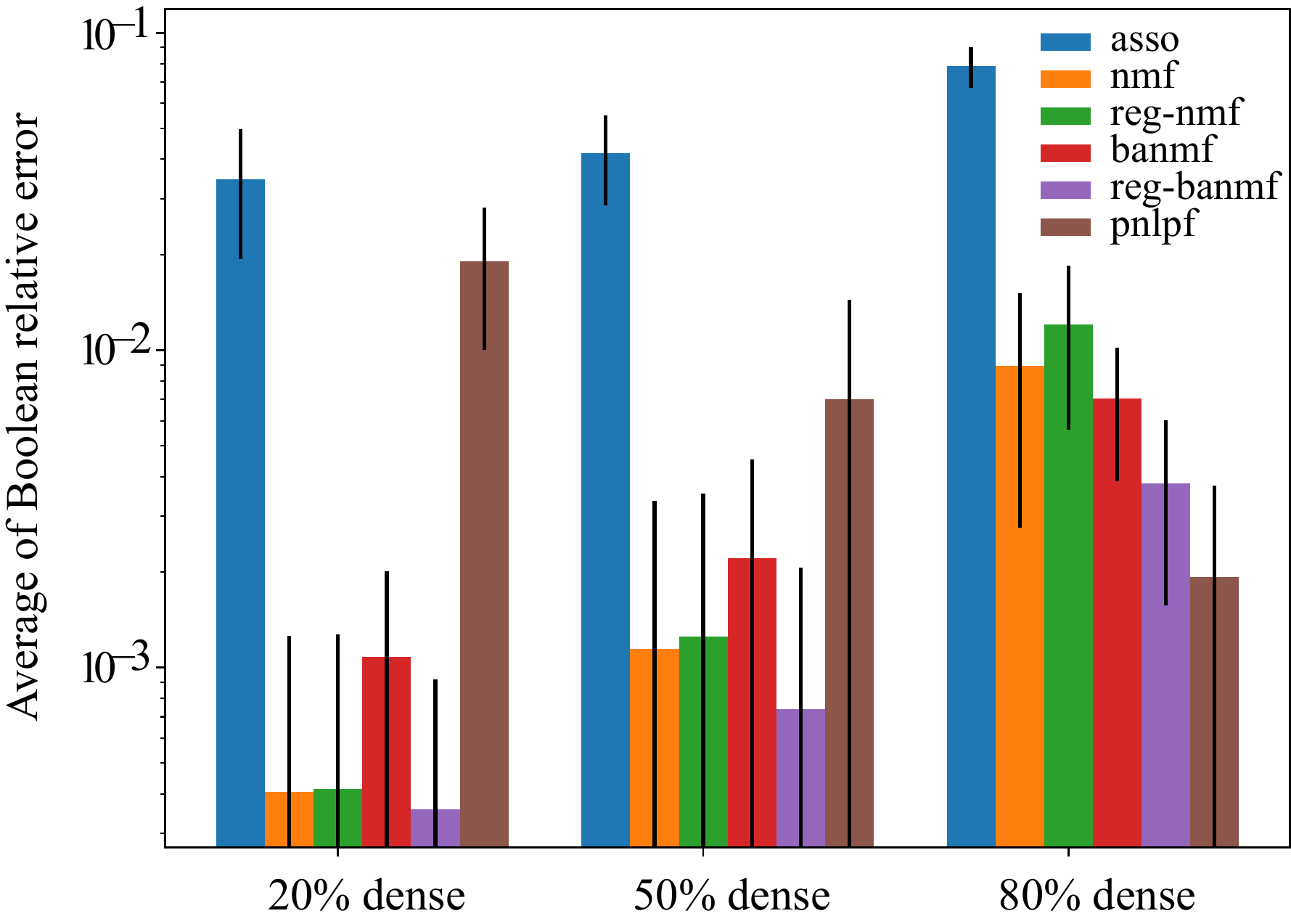}
	\caption{Performance on synthetic datasets at different density levels.}
	\label{fig:synthetic_Bool_err}
\end{figure}

\subsection{Different noise levels} 
	\begin{figure}
		\centering
		\includegraphics[width=0.4 \textwidth]{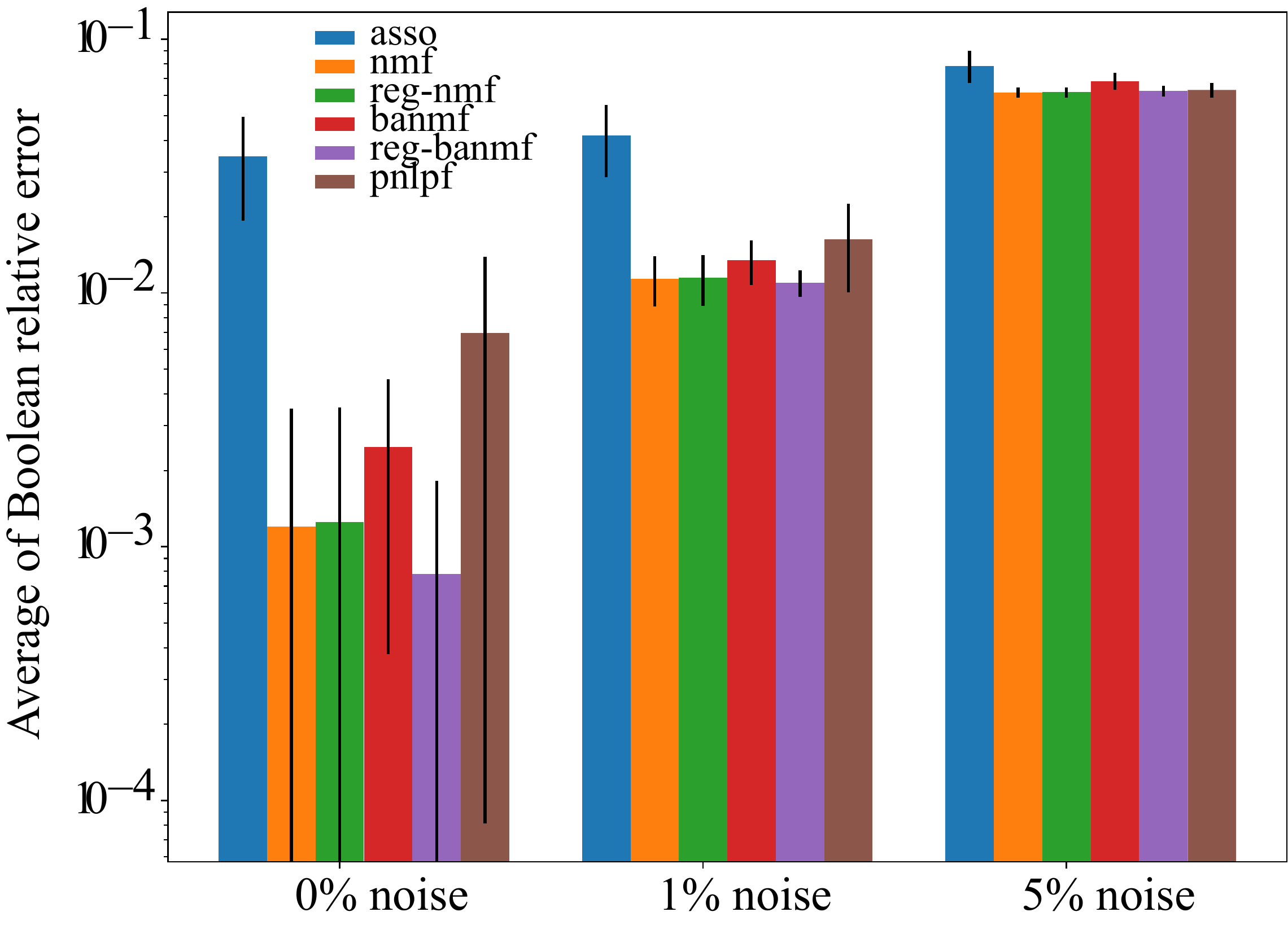}
		\caption{Mean and standard deviation of Boolean relative error of four methods under three noise levels across 100 datasets. (Regularized) \BMF{} are consistently outperform ASSO and PNL-PF.}
		\label{fig:noise_examples}
	\end{figure}

	The next experiment is to investigate the performance of all methods under the effect of noise. Here 100 different $50 \times 50$ matrices with $50\%$ density and latent space $k=5$ are generated with no noise.  
	Then Bernoulli noise is added to the synthetic data according to different noise thresholds. Concretely, 
	\[X^{50 \times 50} = W^{50 \times 5} \bmp_B H^{ 5 \times 50} +_f E
	\]
	where
	\[W_{ij}, H_{ij} \sim Bernoulli(p) \quad E_{i,j} \sim Bernoulli(p_E)\]
	and the flipping operation \(+_f\) is defined as
	\begin{equation}
		X_{ij} = \begin{cases}
			X_{ij} \quad & \text{if } E_{ij} = 0\\
			\neg X_{ij} \quad & \text{if } E_{ij} = 1.\\
		\end{cases}
	\end{equation}
	Each of the 100 matrices dataset is corrupted with  noise levels (\(0\%, 1\%\) and \(5\%\)).   \autoref{fig:noise_examples} shows the mean and standard deviation of relative error across the 100 datasets, in which regularized \BMF{}  consistently outperform ASSO and PNLPF. Moreover, \BMF{} outperforms PNLPF at lower noise thresholds. 

	\subsection{\BMF{} and NMF on different-rank data}
	\label{rank difference subsection}
	The NMF algorithm searches for a best nonnegative rank approximation to $X=WH$. With \BMF{} being algorithmically similar to NMF, we are naturally interested in the performance when there is a gap between the Boolean rank and the nonnegative rank.  Theoretically, it is known that the nonnegative rank is greater than or equal to the Boolean rank \cite{desantis2020factorizations}. To empirically investigate the effect of a gap between the ranks, we generate a suite of Boolean matrices where a gap between $\rank_+(X)$ and $\rankb(X)$ exists. This is done as follows.

	For each $N,M = 10,11, \dots, 50$, $k=2,3,\dots 6$, and density levels $25\%, 50\%,75\%$, we generate 5 Boolean matrices.  We then want to measure the difference $\rank_+(X) - \rankb(X)$ for each Boolean matrix $X$.  Unfortunately, computing either the nonnegative or Boolean rank is NP hard.  However we can efficiently estimate a lower bound on the rank gap.  Indeed by construction, $\rankb(X) \leq k$ and $\rank(X)\leq \rank_+(X)$.  Thus if $k\leq \rank(X)$, then 
	\[
	\rank_+(X) - \rankb(X) \geq \rank(X) - k.
	\]
	Thus for each generated matrix, we check that $k \leq \rank(X)$. If this check fails, we generate a new Boolean matrix with the same parameters. We then compute $\rank(X) - k$ to estimate the lower bound on the gap. 

	Figure~\ref{fig:mc_vs_nmf} depicts the average error of each method as a function of the lower bound estimation of the rank gap. For lower rank gap, many of the methods are comparable. However the error in using NMF grows as the gap between the nonnegative and Boolean ranks grows, while the regularized \BMF{} algorithm is constant and low, and the unregularized \BMF{} falls somewhere in between. The regularized NMF's error also grows with the rank gap, while PNLPF seems to stay constant with a higher error level compared to regularized BANMF. 

	\begin{figure}
		\centering
		\includegraphics[width = 0.6 \textwidth]{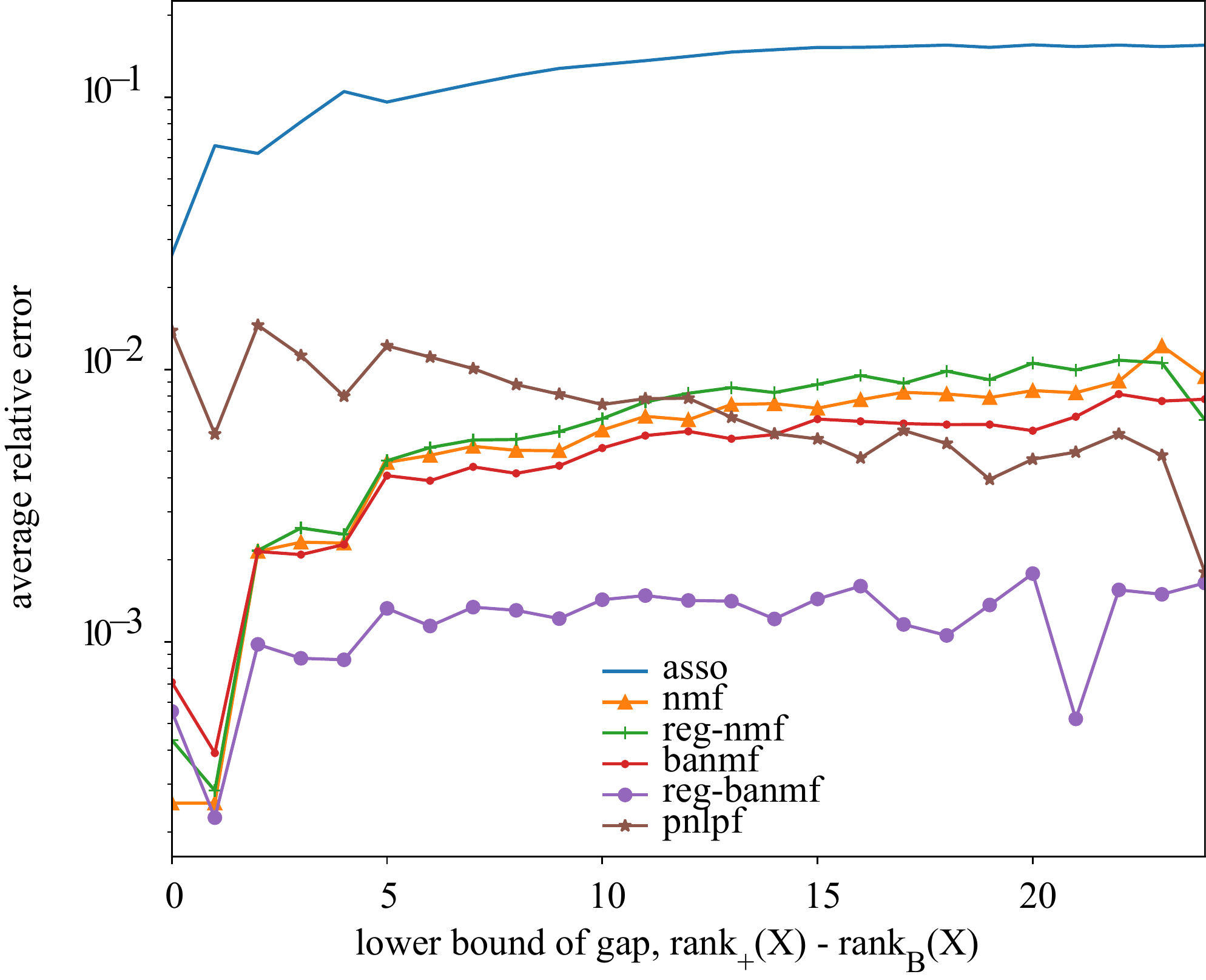}
		\caption{Performance of algorithms on Boolean/nonnegative rank gap dataset. }
		\label{fig:mc_vs_nmf}
	\end{figure}

\subsection{Execution time}
\label{subsection execution time}
	Since PNLPF's performance is comparable with \BMF{} in some previous cases, we further compare the complexity between PNLPF and \BMF{}.	Given the high level of nonlinearity of PNLPF, it is expected to be more computationally expensive than \BMF{} algorithm. 
	In the first comparison, the execution times of 1000 iterations for datasets of different data sizes are measured. \autoref{fig:exec time ratio} shows the execution time in log-scale of PNLPF and \BMF{}. The results show that when the data size is getting larger, PNLPF is increasingly more expensive. For example, at the data size of \(500 \times 500\), the execution time of PNLPF is about \(2^4 = 16\) times longer than of \BMF{}. 
	In the second comparison, \autoref{fig:exec time convg} shows the time that each method converges in Boolean error rate for different datasets. The result shows that as the data size increases, PNLPF takes a longer time to converge. In the conclusion, \BMF{} is less computationally expensive compared to PNLPF.
\begin{figure}
	\centering
	\includegraphics[width = 0.4 \textwidth]{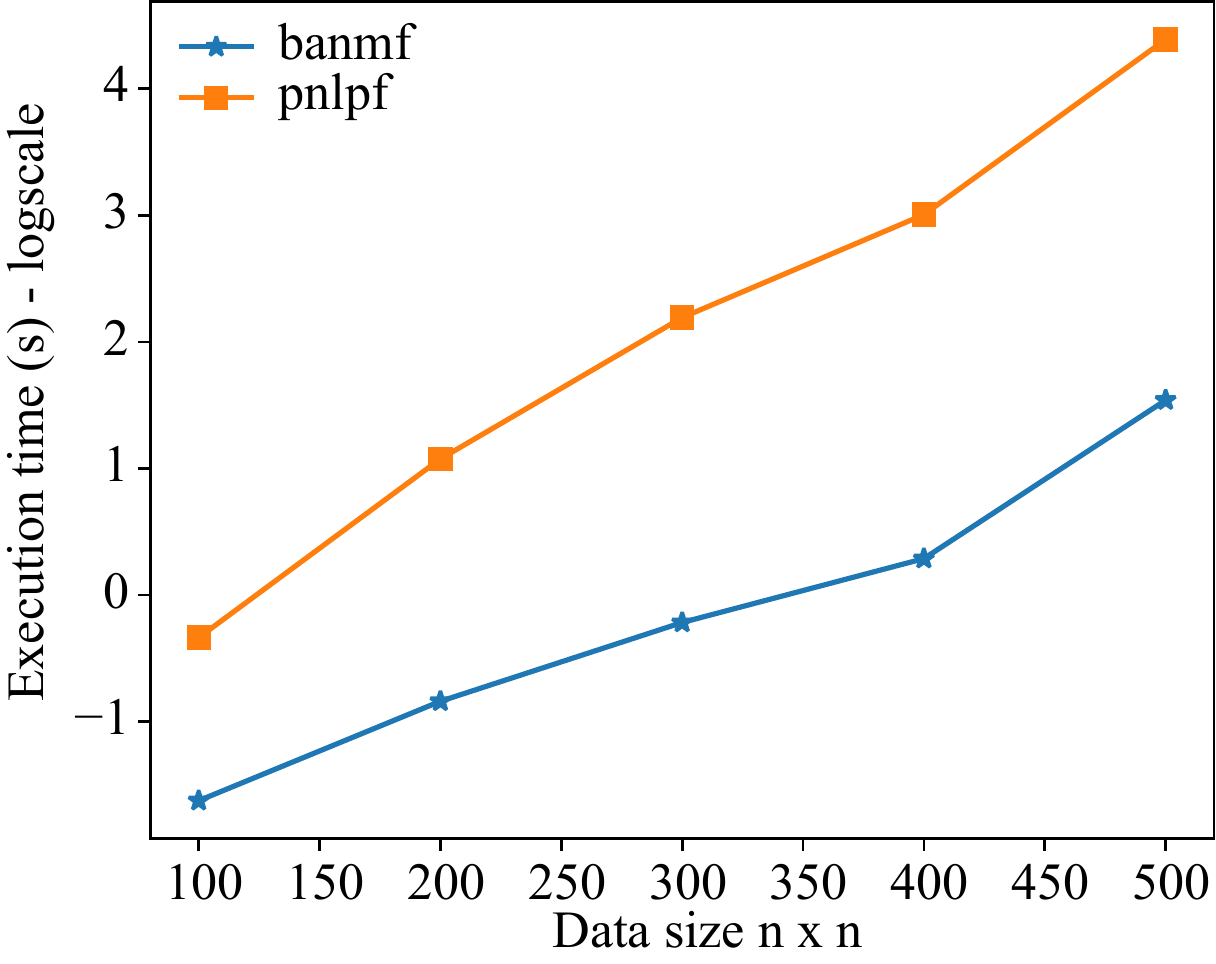}
	\caption{Execution time in log-scale of \BMF{} and PNLPF for 5 different size datasets. Overall, PNLPF is more expensive than \BMF{}, and even more expensive when the data size is larger.}
	\label{fig:exec time ratio}
\end{figure}

\begin{figure}
	\centering
	\includegraphics[width = 0.6 \textwidth]{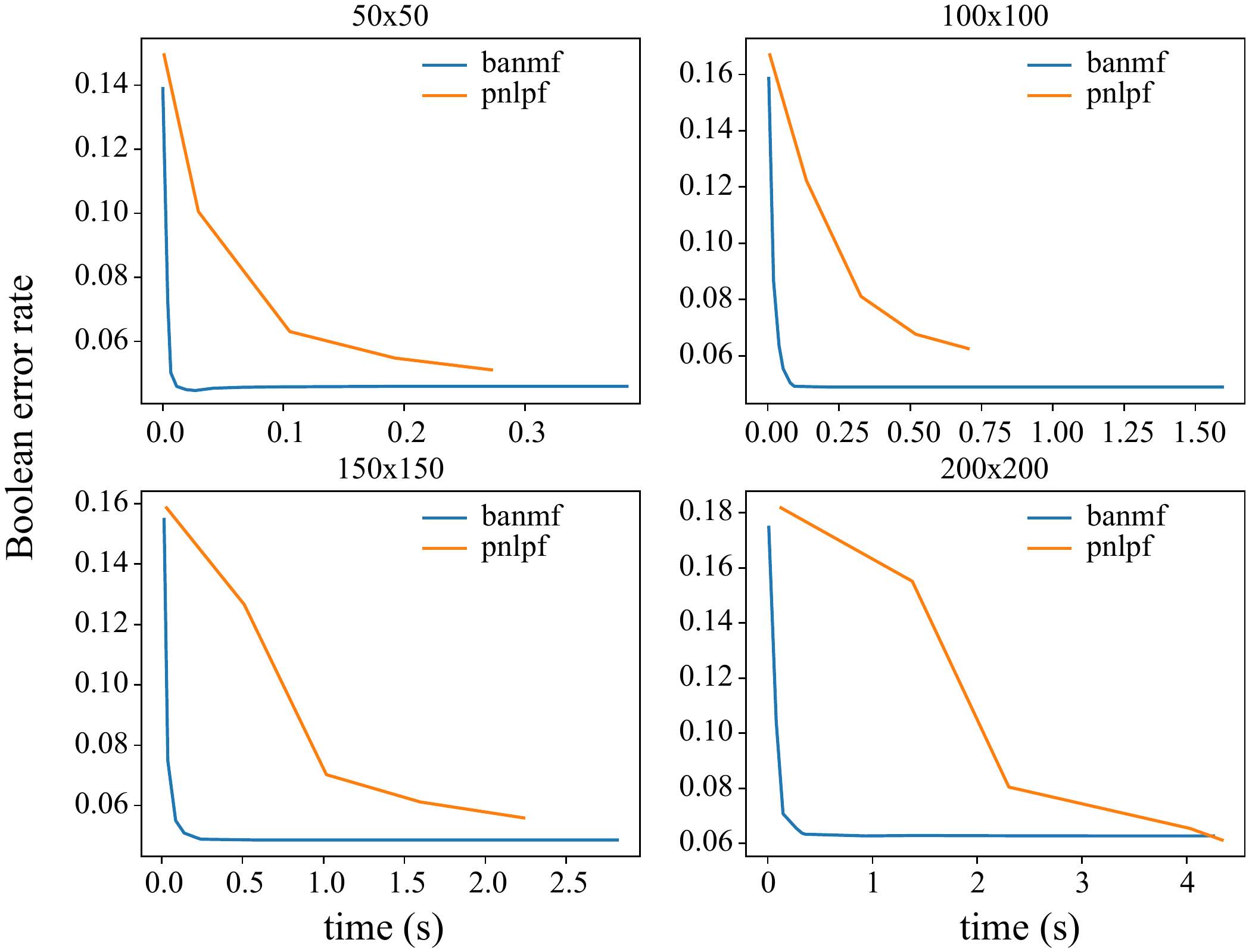}
	\caption{Convergence time in Boolean error rate of BANMF and PNLPF. As the data is larger, PNLPF takes increasingly longer time to converge.}
	\label{fig:exec time convg}
\end{figure}

\subsection{Application in real data}
In this section, we demonstrate the performance of (regularized) \BMF{} and the other methods  on three real datasets: UCI Zoo, Congressional voting records and Lung cancer with probes LUCAP0. 


\subsubsection{UCI zoo dataset}
We first study the performance of BNMF-MC on decomposing the UCI zoo dataset \footnote{https://archive.ics.uci.edu/ml/datasets/Zoo} \cite{Dua:2019}. This dataset consists of 101 animals and 14 binary features, which results in a \(101 \times 14\) Boolean matrix. \autoref{fig:uci_zoo_Bool} shows that regularized \BMF{} and NMF have the best performance, especially with higher latent dimensions. Furthermore, \autoref{fig:uci_zoo_heatmap} shows the original data and reconstructed data from rank-3 decompositions. All methods seem to capture the Boolean structure of the data.

\begin{figure}
	\centering
	\includegraphics[width = 0.5 \textwidth]{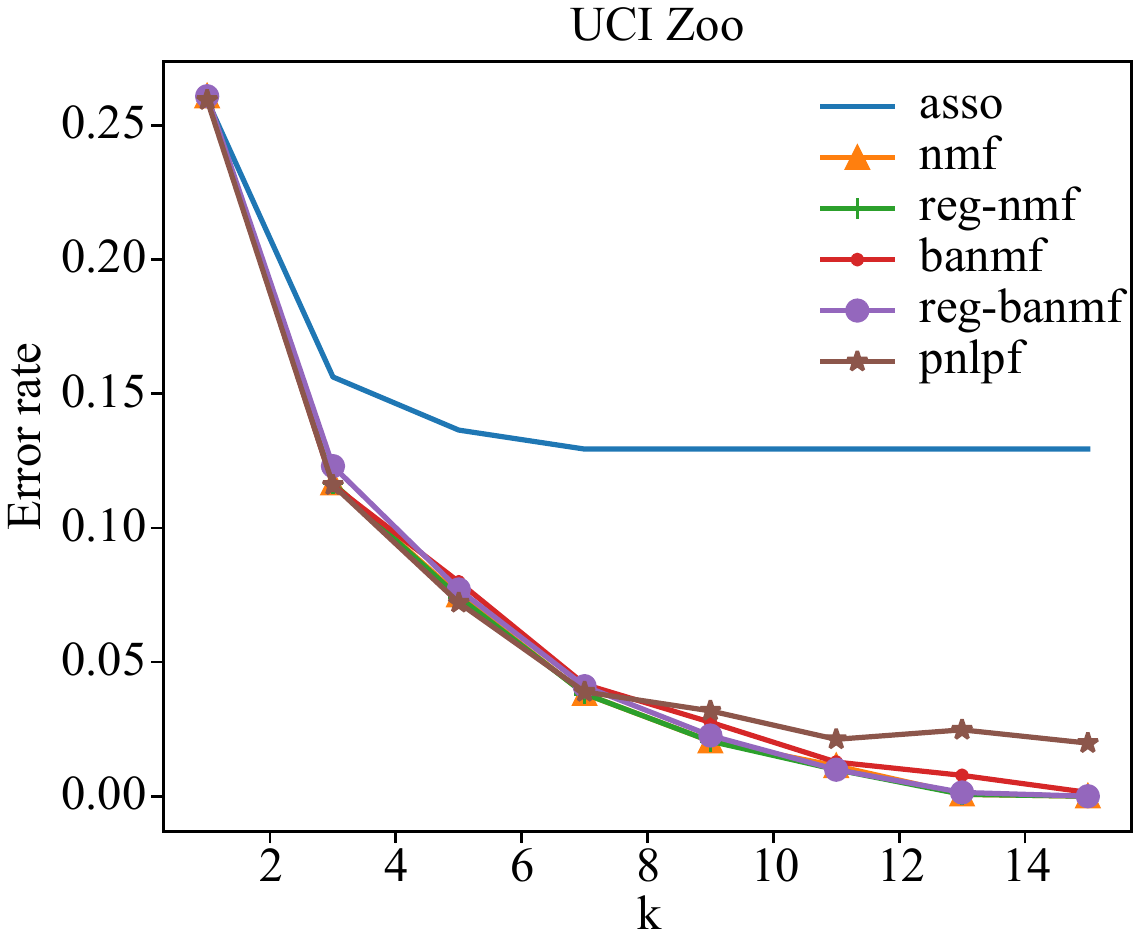}
	\caption{UCI Zoo dataset. Regularized BANMF and NMF have the best performances, especially with higher latent dimensions.}
	\label{fig:uci_zoo_Bool}
\end{figure}

\begin{figure}
	\centering
	\includegraphics[width = 0.5 \textwidth]{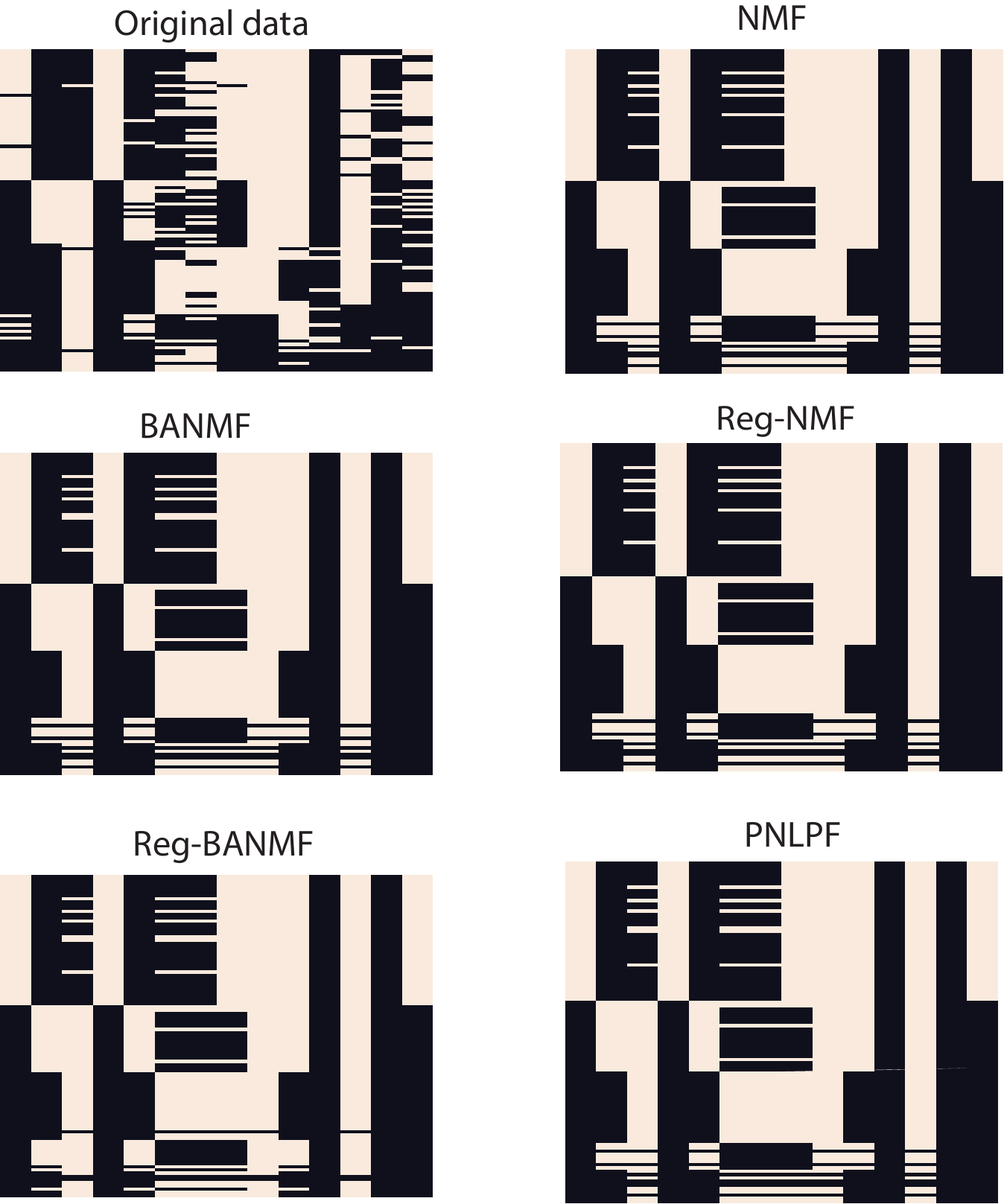}
	\caption{Original data versus reconstructed data from a rank-3 decomposition of (regularized) BANMF and other methods. All methods seem to capture the Boolean structure of the data.}
	\label{fig:uci_zoo_heatmap}
\end{figure}

\subsubsection{Congressional voting records dataset}
The second experiment uses the Congressional Voting Records Dataset from UCI dataset \cite{Dua:2019}. This data includes the votes of the U.S. House of Representatives Congressmen on the 16 key votes. After removing samples with missing data, the Boolean data matrix has the dimension of \(232 \times 16\). \autoref{fig:housevotes} shows that (regularized) NMF are better with (regularized) \BMF{} are slightly behind.

\begin{figure}
	\centering
	\includegraphics[width = 0.5 \textwidth]{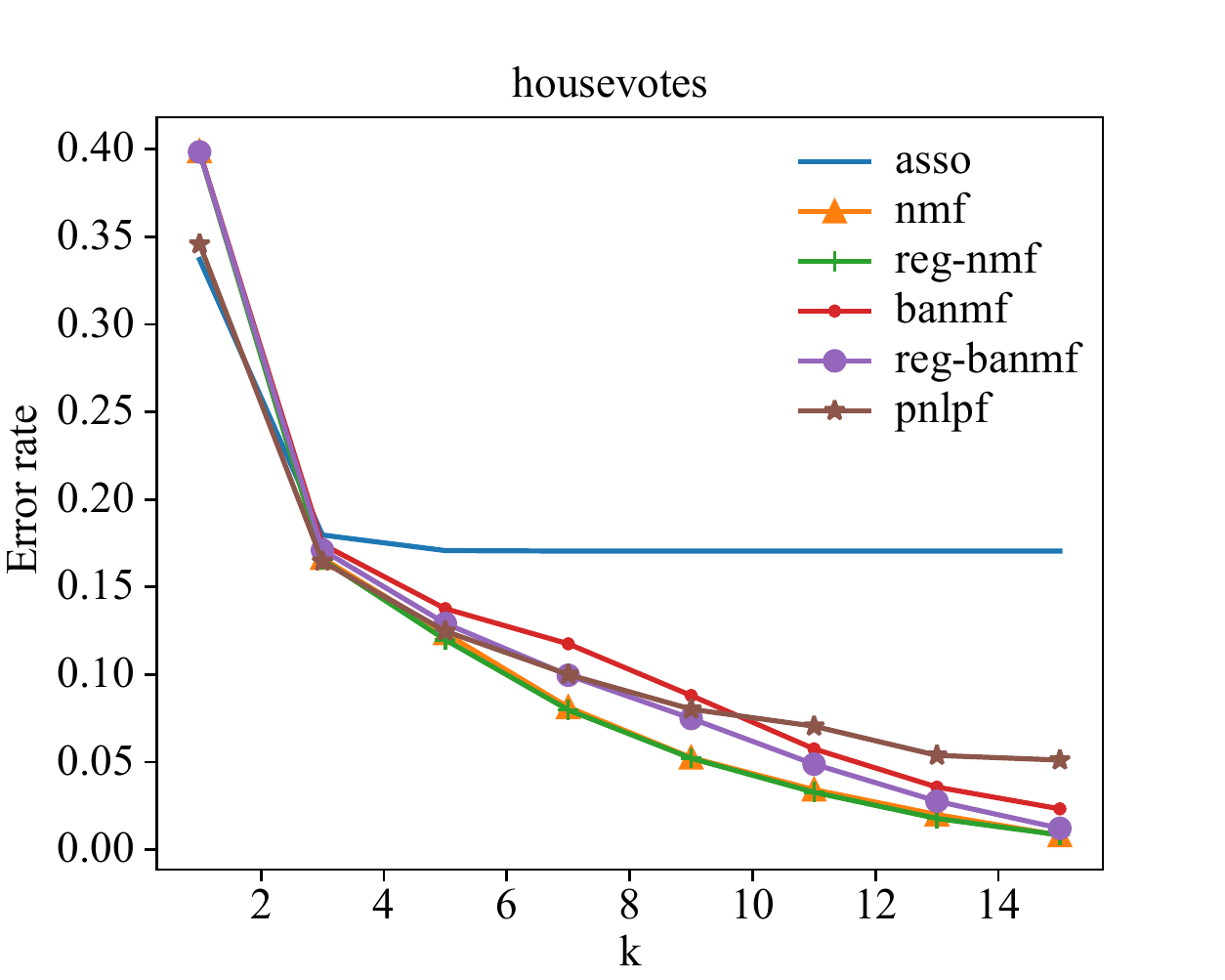}
	\caption{Congressional voting records dataset - (regularized) NMF has the lowest errors with (regularized) \BMF{} are slightly behind.}
	\label{fig:housevotes}
\end{figure}

\subsubsection{Lung cancer with probes dataset (LUCAP0)}
The third data set is the LUCAP0 (Lung Cancer with Probes) data which contains toy data generated artificially by causal Bayesian networks with binary variables, which being used in \cite{miron178boolean}. This model is a medical application for diagnosing lung cancer.\footnote{http://www.causality.inf.ethz.ch/data/LUCAS.html}. The dimensions of the data are \(2000 \times 143\). For this dataset, \autoref{fig:lucap0} shows that (regularized) NMF are slightly better here. Additionally, PNLPF performs well with low latent dimension, but is outperformed by regularized \BMF{} when the latent dimension gets larger. This PNLPF's behavior can also be observed from previous datasets. Moreover, given this size of the data, PNLPF is more computationally expensive than \BMF{} and NMF (see \autoref{subsection execution time} for the execution time comparison.)

\begin{figure}
	\centering
	\includegraphics[width = 0.5 \textwidth]{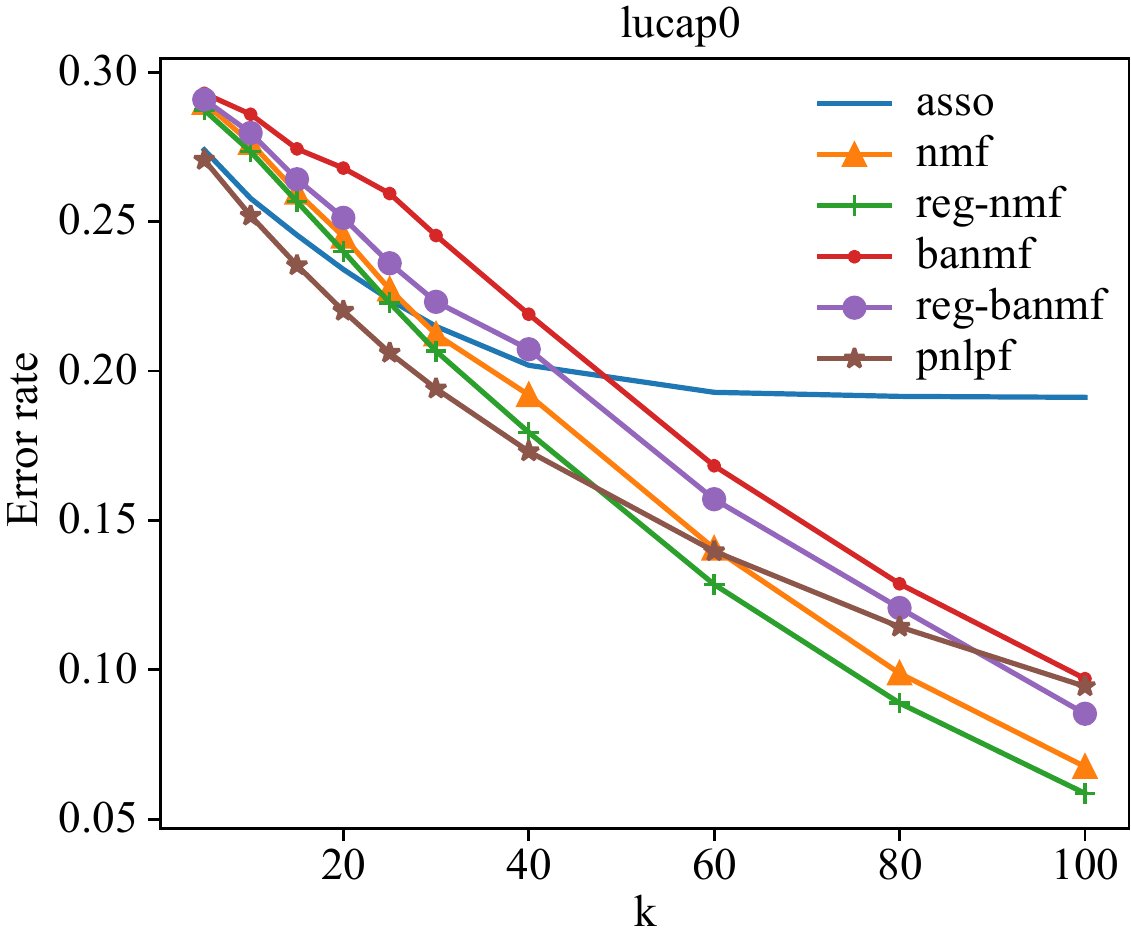}
	\caption{Lung cancer with probes dataset. (Regularized) NMF has the lowest error with higher latent dimension.}
	\label{fig:lucap0}
\end{figure}

\section{Discussion}
In this paper, we propose an approach to solve the Boolean matrix factorization problem using a nonnegative optimization with the constraint over an auxiliary matrix. With this constraint, the algorithm can overcome the difference between Boolean algebra and normal algebra to capture the Boolean structure of the data. The solution is then thresholded to give a solution for BMF problem. We show that the exact solution spaces of the two problem are equivalent, and that the \BMF{} algorithm has the nonincreasing property. The (regularized) \BMF{} algorithm is shown to be on outperform other methods in experiments on synthetic examples at particular density, and noisy regimes.  Furthermore, these algorithms prove to be superior methods when  there is a gap between Boolean rank and nonnegative rank. Outside these regimes, the (regularized) \BMF{} are still competitive with other state of the art algorithms.


\bibliography{main.bib}

\begin{thebibliography}{10}

\bibitem{belohlavek2010discovery}
Radim Belohlavek and Vilem Vychodil.
\newblock Discovery of optimal factors in binary data via a novel method of
  matrix decomposition.
\newblock {\em Journal of Computer and System Sciences}, 76(1):3--20, 2010.

\bibitem{desantis2020factorizations}
Derek DeSantis, Erik Skau, and Boian Alexandrov.
\newblock Factorizations of binary matrices--rank relations and the uniqueness
  of boolean decompositions.
\newblock {\em arXiv preprint arXiv:2012.10496}, 2020.

\bibitem{Dua:2019}
Dheeru Dua and Casey Graff.
\newblock {UCI} machine learning repository, 2017.

\bibitem{everett2013introduction}
B~Everett.
\newblock {\em An introduction to latent variable models}.
\newblock Springer Science \& Business Media, 2013.

\bibitem{lee1999learning}
Daniel~D Lee and H~Sebastian Seung.
\newblock Learning the parts of objects by non-negative matrix factorization.
\newblock {\em Nature}, 401(6755):788--791, 1999.

\bibitem{lee2001algorithms}
Daniel~D Lee and H~Sebastian Seung.
\newblock Algorithms for non-negative matrix factorization.
\newblock In {\em Advances in neural information processing systems}, pages
  556--562, 2001.

\bibitem{lucchese2010mining}
Claudio Lucchese, Salvatore Orlando, and Raffaele Perego.
\newblock Mining top-k patterns from binary datasets in presence of noise.
\newblock In {\em Proceedings of the 2010 SIAM International Conference on Data
  Mining}, pages 165--176. SIAM, 2010.

\bibitem{miettinen2008discrete}
Pauli Miettinen, Taneli Mielik{\"a}inen, Aristides Gionis, Gautam Das, and
  Heikki Mannila.
\newblock The discrete basis problem.
\newblock {\em IEEE transactions on knowledge and data engineering},
  20(10):1348--1362, 2008.

\bibitem{miron178boolean}
Sebastian Miron, Mamadou Diop, Anthony Larue, Eddy Robin, and David Brie.
\newblock Boolean decomposition of binary matrices using a post-nonlinear
  mixture approach.
\newblock {\em Signal Processing}, 178:107809, 2021.

\bibitem{Monson1995}
S.~D. Monson, N.~J. Pullman, and R.~Rees.
\newblock {A Survey of Clique and Biclique Coverings and Factorizations of
  (0,1)-Matrices}.
\newblock {\em Bull. ICA}, 14:17--86, 1995.

\bibitem{ravanbakhsh2016boolean}
Siamak Ravanbakhsh, Barnab{\'a}s P{\'o}czos, and Russell Greiner.
\newblock Boolean matrix factorization and noisy completion via message
  passing.
\newblock In {\em ICML}, volume~69, pages 945--954, 2016.

\bibitem{rukat2017bayesian}
Tammo Rukat, Chris~C Holmes, Michalis~K Titsias, and Christopher Yau.
\newblock Bayesian boolean matrix factorisation.
\newblock {\em arXiv preprint arXiv:1702.06166}, 2017.

\bibitem{spearman1961general}
Charles Spearman.
\newblock ``{G}eneral intelligence,'' objectively determined and measured.
\newblock {\em The American Journal of Psychology}, 15, 1961.

\bibitem{Stewart1993}
G.~W. Stewart.
\newblock On the early history of the singular value decomposition.
\newblock {\em SIAM review}, 35(4):551--566, 1993.

\bibitem{zhang2007binary}
Zhongyuan Zhang, Tao Li, Chris Ding, and Xiangsun Zhang.
\newblock Binary matrix factorization with applications.
\newblock In {\em Seventh IEEE international conference on data mining (ICDM
  2007)}, pages 391--400. IEEE, 2007.

\end{thebibliography}
\bibliographystyle{plain}

\end{document}